\documentclass[10pt,draftcls,onecolumn]{IEEEtran}
\usepackage{graphicx}
\usepackage{epstopdf}
\usepackage{multicol}
\usepackage{stfloats}
\usepackage{amsfonts}
\usepackage{mathrsfs}
\usepackage{multicol}
\usepackage{stfloats}
\usepackage{cite}
\usepackage{amsfonts}
\usepackage{amsthm}
\usepackage{mathrsfs}
\usepackage{amsfonts}
\usepackage{empheq}
\usepackage{color}
\usepackage[ruled]{algorithm2e}
\usepackage{fancybox}
\usepackage{framed}
\usepackage{amsfonts, color}
\usepackage{color}
\usepackage{setspace}
\usepackage[dvips]{epsfig}
\usepackage{psfrag, amssymb, amsmath, cite}
\usepackage[colorlinks,linkcolor=red,anchorcolor=blue,citecolor=blue]{hyperref}
\usepackage{verbatim}
\newtheorem{theorem}{\textbf{Theorem}}
\newtheorem{assumption}{\textbf{Assumption}}

\newtheorem{lemma}{\textbf{Lemma}}

\newtheorem{corollary}{\textbf{Corollary}}
\newtheorem{remark}{\textbf{Remark}}

\renewcommand\footnoterule{\kern-3pt \hrule width 2in \kern 2.6pt}

\IEEEoverridecommandlockouts

\title{\bf  \LARGE Dimensional-invariance principles in coupled dynamical systems:  A unified analysis and applications}

\author{Zhiyong Sun and Changbin Yu
\thanks{This work was supported by
the Australian Research Council under grant
 DP130103610 and DP160104500.}
\thanks{Z. Sun and C. Yu are with
Research School of Engineering, The Australian National University, Canberra ACT
2601, Australia. C. Yu is also with  Hangzhou Dianzi University, Hangzhou, China. Emails: \ {\tt \small  sun.zhiyong.cn@gmail.com, brad.yu@anu.edu.au}  }
}

\begin{document}
\maketitle
\thispagestyle{empty}
\pagestyle{empty}

\begin{abstract}
In this paper we study coupled dynamical systems and {\color{blue} investigate  dimension properties of  the subspace spanned by solutions of each individual system.   Relevant problems on \textit{collinear dynamical systems} and their variations are discussed recently by Montenbruck et. al. in \cite{collinear2017ACC}, while in this paper we aim to provide a unified analysis to derive the dimensional-invariance principles for networked coupled systems, and to generalize the invariance principles for networked systems with more general forms of coupling terms}. To be specific, we consider two types of coupled systems, one with scalar couplings and the other with matrix couplings. Via the \textit{rank-preserving flow theory}, {\color{blue} we show that any scalar-coupled dynamical system (with constant, time-varying or state-dependent couplings)} possesses the  dimensional-invariance principles, in that the dimension of the subspace spanned by the individual systems' solutions remains invariant.   For coupled dynamical systems with matrix coefficients/couplings,  necessary and sufficient conditions {\color{blue} (for constant, time-varying and state-dependent couplings) } are given to characterize dimensional-invariance principles. {\color{blue}The proofs via a rank-preserving matrix flow theory in this paper simplify the analysis in \cite{collinear2017ACC}, and we also extend  the invariance principles to the  cases of time-varying couplings and state-dependent couplings. Furthermore, subspace-preserving property and signature-preserving flows are also developed for coupled networked systems with particular coupling terms. } These invariance principles provide insightful characterizations to analyze  transient behaviors and solution evolutions  for a large family of coupled systems, such as multi-agent consensus dynamics, distributed coordination systems, formation control systems, among others.
\end{abstract}


\section{Introduction}
In this paper we consider the following coupled dynamical systems consisting of $n$ individual systems
\begin{align} \label{eq:general_system}
\dot x_i(t) = \sum_{i =1}^{n} \kappa_{ij}(t) x_j(t),
\end{align}
where $x_i \in \mathbb{R}^d$ is the state for system $i$, $\kappa_{ij}$ is  a coupling weight scalar or matrix between systems $j$ and $i$ (when $i = j$, then it is a coefficient scalar/matrix $\kappa_{ii}$ for system $i$). The coupling/coefficient weights  could be constant,   time-varying, or state-dependent.
The system \eqref{eq:general_system} serves as a very general model to describe many different types of coupled/networked systems, such as formation control systems \cite{anderson2014counting, krick2009stabilisation, sun2016exponential}, network computation systems \cite{costello2014degree, costello2015global},  multi-agent consensus dynamics \cite{olfati2004consensus}, \cite{ren2007information}.

Coupled dynamical systems are often operated in a networked manner, where each individual system interacts with other systems to perform a global or common task. Networked control systems  in the general model \eqref{eq:general_system} have been attracting increasing attention in the recent decade and can be found in a variety of applications.  Depending on the actual control task, the coupling terms can be designed to reflect  information flows between spatially distributed systems,  communication requirements or constraints, or cooperative  interactions incorporating local tasks to achieve a global task \cite{ren2007information}.

{\color{blue} This paper builds on recent work of \cite{sun2015rigid} and \cite{collinear2017ACC}, and aims to identify several invariance principles (some were developed in \cite{collinear2017ACC}) for the solutions of each   system  arising from different couplings and interactions between individual systems.
Invariance principles for distance-based formation systems \cite{krick2009stabilisation}
(which can be described as a special case of \eqref{eq:general_system}) were discussed in \cite{sun2015rigid}, which show that all individual agents' solutions span a linear subspace with a constant dimension over time. A more recent paper \cite{collinear2017ACC} provided a comprehensive study on this invariance principle for coupled linear systems, by starting with the collinear dynamical systems with constant coefficients and couplings. }
In this paper, we discuss two types of coupled dynamical systems that can be represented by \eqref{eq:general_system}, one with scalar couplings, and the other with matrix couplings, respectively. These principles relate to the invariance of the dimensions of the subspaces spanned by the solutions of each individual system, which are thus termed as \emph{dimensional-invariance principles}. 
{\color{blue} We aim to provide a unified analysis to establish such invariance principles, by using a different approach from \textit{rank-preserving matrix flow theory} which simplifies the analysis in  \cite{collinear2017ACC}. Furthermore, as compared to \cite{collinear2017ACC} which discussed networked systems with constant couplings,  we also provide several generalizations of the dimensional-invariance principle  to more general coupling terms (that include time-varying and state-dependent couplings) and to more elaborated invariance principles such as subspace-preserving property and signature-preserving principles (their definitions will become clear in the context in later sections).   }


These invariance principles   are fundamental yet universal properties for coupled dynamical systems. We note that in most papers on coupled/networked control systems, the focus has been on the stability and convergence analysis, while \emph{transient behaviors} are largely ignored. The results revealed from the dimensional-invariance principles provide us with additional insights on the transient behaviors and evolutions of all individual solutions, and could  assist the convergence and stability analysis of the overall coupled dynamical systems.
An example is the distance-based formation control system   described by gradient flows from  potential functions of interest, which show that an initially collinear formation remains collinear for all time under such flows (see e.g. \cite{krick2009stabilisation}).

The invariance principles also provide feasible coupling conditions to guarantee that  the solutions of individual systems are constrained in some smaller dimensional spaces, which could find particular applications in several practical scenarios. For example,   \emph{collinear} solutions of a coupled dynamical system are of particular interests. In \cite{anderson2014counting}, a line formation, in which individual systems' states are confined in a 1-D subspace, is studied with insights to more general formations on other dimensions. As another example, for a coupled dynamical system that describes the coordination of  multiple mobile antennas, collinear solutions have practical significance to align   directions of all antennas in a single line \cite{stutzman2012antenna}. Motivated by these practical applications, the theory of collinear dynamical systems was studied in \cite{collinear2017ACC}. The dimensional-invariance principles, {\color{blue} established in \cite{sun2015rigid} (for multi-agent formation systems) and \cite{collinear2017ACC} (for constant couplings) and generalized in this paper (for general couplings)}, will provide an insightful framework to facilitate these applications.

The remaining parts of this paper are structured as follows. In Section  \ref{sec:scalar-weighted}, we prove that the  solutions of scalar-coupled dynamical systems have the dimensional-invariance (and furthermore, subspace-preserving) principles, {\color{blue}which thus generalizes the results in \cite{collinear2017ACC}}.  In Section  \ref{sec:matrix-weighted}, matrix-coupled dynamical systems are discussed, for which necessary and sufficient conditions are given to guarantee the dimensional-invariance principle. {\color{blue} The subtle difference between the subspace-preserving property and (the more general) dimensional-invariance principle is also elaborated in this section.} Applications of the invariance principles in {\color{blue} general} formation control systems are shown in Section  \ref{sec:application_formation}. Section   \ref{sec:conclusions} presents the conclusions of this paper. In the appendix sections, we present preliminary background on rank-preserving flows, some extensions and proofs, and a brief review of several popular networked dynamical systems that fit in the general model \eqref{eq:general_system}.

\subsection{Assumptions and solution issues of \eqref{eq:general_system}}
To address the solution issue of the coupled dynamical system \eqref{eq:general_system}, we impose the following mild assumption.
\begin{assumption}
The coefficient/coupling terms $\kappa_{ij}$ are continuous scalar/matrix functions.
\end{assumption}

The above mild assumption guarantees the existence and uniqueness of the  solutions for system \eqref{eq:general_system} \cite[Chapter 1.2]{brockett1970finite}.  Note that we do not impose additional assumptions on $\kappa_{ij}$'s. They can be constant, time-varying, state-dependent or other general continuous functions. {\color{blue} Note also that the system \eqref{eq:general_system} can be  a coupled \textit{time-invariant linear} system (when $\kappa_{ij}$ is constant),  a  coupled \textit{time-varying linear} system (when $\kappa_{ij}(t)$ is time-varying), or a coupled \emph{nonlinear} system (when the coupling term $\kappa_{ij}$  depends on the state $x$). For example,  in multi-agent formation and swarm control, the coupling term $\kappa_{ij}$ is usually a function of system states $x$, written as $\kappa_{ij}(x)$. (See detailed expressions in Table \ref{table_wij} in the Appendix.)   An example on distance-based formation control systems described by \eqref{eq:general_system} in which $\kappa_{ij}$ is a function of $x$ will be discussed in Section  \ref{sec:application_formation}. }

%
%

\section{Coupled dynamical systems with scalar-weighted couplings} \label{sec:scalar-weighted}
Consider the following coupled dynamical systems with scalar couplings
\begin{align}  \label{eq:coupled_linear_wij}
\dot x_i(t) = \sum_{i =1}^{n} w_{ij}(t) x_j(t),
\end{align}
where $w_{ij}$ is a scalar (constant or time-varying) coupling weight between agents $j$ and $i$. Note that we do not require $w_{ij} = w_{ji}$, i.e., the coupling weight could be \emph{asymmetric}.

\subsection{Main results}
In this section we show that the coupled dynamical system \eqref{eq:coupled_linear_wij} has the following dimensional-invariance principle.
\begin{theorem} \label{thm_rankpreserving_scalar}
The coupled dynamical system  \eqref{eq:coupled_linear_wij} has the dimensional-invariance principle in the sense that
\begin{empheq}[box=\fbox]{align} \label{eq:rankpreserving_sense}
\text{rank}(X(t)) = r, (r \leq d), & \,\,\,\, \forall t \geq 0,  \nonumber \\
&  \text{if} \,\,\,\,  \text{rank}(X(0)) = r.
\end{empheq}
where $X = [x_1, x_2, \ldots, x_n] \in \mathbb{R}^{d \times n}$.
\end{theorem}

\begin{proof}
Define the composite vector $x = [x_1^{\top}, x_2^{\top}, \ldots, x_n^{\top}]^{\top} \in \mathbb{R}^{dn}$. In order to obtain a compact form of the system $\dot x$, we define the matrix $W(t) = \{w_{ij}(t)\} \in \mathbb{R}^{n \times n}$. Therefore, a compact form of \eqref{eq:coupled_linear_wij} can be  written as
\begin{align}  \label{eq:coupled_linear_wij_compact}
\dot x(t) = (W(t) \otimes I_d) x(t),
\end{align}
where $\otimes$ denotes the Kronecker product. The vector differential equation \eqref{eq:coupled_linear_wij_compact} on the real vector space $\mathbb{R}^{dn}$ can be stated equivalently  as the following differential  flow on the matrix space $\mathbb{R}^{d \times n}$ (without involving  the Kronecker product term)
\begin{align}  \label{eq:compact_matrix_wij}
\dot X(t) = X(t) W^{\top}(t).
\end{align}
Since the solution of \eqref{eq:coupled_linear_wij} is well defined, the existence and uniqueness of the solution to \eqref{eq:compact_matrix_wij} is also well guaranteed.
Then according to Lemma \ref{lemma_rank} (in Appendix I), the rank-preserving property  of the matrix flow \eqref{eq:compact_matrix_wij} follows   by observing $B(t) = W^{\top}(t) $ and $A(t) = 0$, which implies the dimensional-invariance property of the solutions to \eqref{eq:coupled_linear_wij} in the sense of \eqref{eq:rankpreserving_sense}.
\end{proof}

{\color{blue}A similar result on the invariance principle of \eqref{eq:coupled_linear_wij} with constant couplings $\omega_{ij}$ was established in \cite{collinear2017ACC} with a different proof, while Theorem~\ref{thm_rankpreserving_scalar} has provided a more general result that  also extends to time-varying couplings $\omega_{ij}(t)$. } We now show a stronger result, that if initial conditions are chosen from some   subspace, then the solutions of the coupled system \eqref{eq:coupled_linear_wij} will always be in that   subspace.
\begin{corollary} \label{corollary_subspace_preserving}
In addition to the rank-invariance principle proved in Theorem~\ref{thm_rankpreserving_scalar}, the solutions of the coupled dynamical systems \eqref{eq:coupled_linear_wij} are subspace-preserving   in the sense that
\begin{empheq}[box=\fbox]{align}  \label{eq:subspacepreserving_sense}
& \text{span}([x_1(t), x_2(t), \ldots, x_n(t)])  \nonumber \\
& = \text{span}([x_1(0), x_2(0), \ldots, x_n(0)]) \,\,\,\, \forall t \geq 0.
\end{empheq}
\end{corollary}

\begin{proof}
The expression of the system equation $\dot X$ in \eqref{eq:compact_matrix_wij} satisfies the matrix differential equation in \eqref{eq:lemma_subspace} with $W(t): = B^\top(t)$. Therefore, {\color{blue} the statement follows from} Lemma \ref{lemma_subspace} (in Appendix II).
\end{proof}
\subsection{Interpretations and implications}
The system \eqref{eq:coupled_linear_wij} describes a very general form of coupled dynamical systems which  encompass  many control systems that have been actively studied in the literature. Examples include the distributed system for networked function computation   \cite{costello2014degree}, \cite{costello2015global},
multi-agent consensus  systems initialled in \cite{olfati2004consensus} (undirected graphs), and developed in e.g. \cite{ren2007information} (directed graphs) and \cite{martin2013continuous} (time-varying couplings), and distributed formation control systems \cite{anderson2014counting, krick2009stabilisation, sun2016exponential}. The results established in this section indicate that, the solutions for individual systems coupled in the form \eqref{eq:coupled_linear_wij} will span a subspace of the same dimension as that spanned by initial conditions, and solutions will be constrained in that subspace over time.

The dimensional-invariance principle for a particular distance-based formation control system \cite{krick2009stabilisation} has been proved in our previous paper \cite{sun2015rigid}. We note that such a principle also holds for a large family of formation control systems, including those covered in \cite{sun2016exponential}.  In a later section we will show, by this example, how this invariance principle could assist our understanding on the evolutions  of agents' positions in a multi-agent formation system.


Table  \ref{table_wij} in the Appendix reviews
several typical coupled dynamical systems  with scalar couplings reported in the literature that can be described by the general form \eqref{eq:coupled_linear_wij}. As a consequence of Theorem \ref{thm_rankpreserving_scalar} and Corollary \ref{corollary_subspace_preserving}, all systems reviewed in Table  \ref{table_wij} satisfy the dimensional-invariance and subspace-preserving property.

\section{Coupled dynamical systems with matrix-weighted couplings}  \label{sec:matrix-weighted}
In this section we consider the following dynamical systems with matrix couplings
\begin{align} \label{eq:coupled_linear_Wij}
\dot x_i(t) &= \sum_{j=1}^n W_{ij}(t) x_j(t)  \nonumber \\
& = W_{i1}(t) x_1(t) + W_{i2}(t) x_2(t) + \cdots + W_{in}(t) x_n(t),
\end{align}
where  $W_{ij}(t) \in \mathbb{R}^{d \times d}$ is the state coefficient or coupling matrix: i.e.,  $W_{ii}(t)$ is the coefficient matrix for system $i$, and $W_{ij}(t)$ is the coupling matrix from systems $j$  to system $i$.

\subsection{Main results}
The main result in this section is the following theorem:
\begin{theorem}  \label{thm:rank_preserving_matrix}
The coupled dynamical systems \eqref{eq:coupled_linear_Wij} have the dimensional-invariance principle in the sense of \eqref{eq:rankpreserving_sense} if and only if the coefficient and coupling matrices $W_{ij}(t)$ satisfy the following condition
\begin{align}
W_{ii}(t) &= A(t)  + b_{ii}(t) I_d, i = 1, 2, \ldots, n;   \label{eq:rank_preserving_matrix1} \\
W_{ij}(t) &= b_{ji}(t) I_d,  i, j =  1, 2, \ldots, n, i \neq j,  \label{eq:rank_preserving_matrix2}
\end{align}
for some matrix $A(t) \in \mathbb{R}^{d \times d}$ and scalars $\{b_{ij}(t)\}, i, j =  1, 2, \ldots, n$.
\end{theorem}

The proof can be found in   Appendix III. {\color{blue} Again, we emphasize that the coupling matrices $W_{ii}$ can be constant or time-varying.} We note that in \cite{collinear2017ACC}, an   equivalent  condition for coupled dynamical systems' solutions to be collinear (or to be $r$-coplanar with dimension $r$) was obtained, via a somewhat more complicated proof. {\color{blue} Here, we provide a unified analysis for the coupling condition via the rank-preserving flow theory, which provides additional insights to the dimensional-invariance principle for matrix-coupled dynamical systems. We have further generalized the results in \cite{collinear2017ACC} from time-invariant couplings to time-varying couplings. }

The following corollary further characterizes the solution property for the case of $n$ coupled dynamical systems.
\begin{corollary}
Suppose the coupled system \eqref{eq:coupled_linear_Wij} consists of $n$ individual systems, and initial conditions $x(0)$ for all the coupled systems are chosen to satisfy  $X(0) \in \mathbb{S}(n)$ (i.e. the real symmetric matrix space).  Then the coupled dynamical systems \eqref{eq:coupled_linear_Wij} have both the dimensional-invariance principle and signature-preserving property \footnote{The signature of a real symmetric matrix refers to the   number (counted with multiplicity) of its positive, negative and zero eigenvalues. See its precise definition in the Appendix. } if and only if the coefficient and coupling matrices $W_{ij}(t)$ satisfy the following condition
\begin{align}
W_{ii}(t) &= A(t)  + a_{ii}(t) I_d, i = 1, 2, \ldots, n;  \label{eq:theorem2_condition1}\\
W_{ij}(t) &= a_{ij}(t) I_d,  i, j =  1, 2, \ldots, n, i \neq j,  \label{eq:theorem2_condition2}
\end{align}
for some matrix $A(t)  = \{a_{ij}(t)\}\in \mathbb{R}^{d \times d}$.
\end{corollary}

\begin{proof}
By invoking Lemma \ref{lemma_rank_symmetric} (in Appendix II), the above condition can be proved by modifying $B$ in the proof of Theorem 2 as $A^\top$.
\end{proof}
\subsection{Interpretations and implications}
The coupled dynamical systems \eqref{eq:coupled_linear_Wij} with matrix weights are also very general that can describe many different types of distributed/networked control systems. Examples include  the matrix-weighted consensus dynamics \cite{tuna2016synchronization}, bearing-based formation control systems  \cite{zhao2016bearing}, or networked linear systems for synchronization \cite{scardovi2009synchronization}.

To guarantee the invariance of the dimensions of the subspaces spanned by individual systems' solutions, the coefficient matrices for each individual system  should have the same matrix structure, with the difference being a scalar {\color{blue} multiple of an identity matrix (hence a diagonal matrix with the same diagonal entry $b_{ii}(t)$)}. Furthermore, the couplings should also be a scalar {\color{blue} multiple of an identity matrix (hence a diagonal matrix with the same diagonal entry $b_{ji}(t)$)}. Since the condition is necessary and sufficient, for other types of couplings between individual systems that are not in the forms of \eqref{eq:rank_preserving_matrix1} and \eqref{eq:rank_preserving_matrix2},  the dimensional-invariance principles cannot be guaranteed.

We also note a difference of the invariance principles between the scalar-coupling case and the matrix-coupling case. As proved in Corollary \ref{corollary_subspace_preserving}, the solutions of coupled system \eqref{eq:coupled_linear_wij} {\color{blue} with scalar couplings} not only span a subspace of the same dimension to that of their initial conditions, but also \textit{evolve in that same subspace spanned by   initial conditions}. However, this subspace-preserving property is not guaranteed for the solutions of the coupled system \eqref{eq:coupled_linear_Wij} {\color{blue} with matrix couplings}. Theorem \ref{thm:rank_preserving_matrix} only shows the invariance of the dimension of the spanned subspace, while the solutions may also evolve in a \textit{different} subspace with the same dimension. {\color{blue} Similarly to \cite{helmke1994optimization} and \cite{collinear2017ACC}}, we introduce the concept of Grassmannian subspace to illustrate the difference.  
The Grassmannian, denoted as $\text{Gr}(r, d)$, is a space which parameterizes all linear subspaces of a given dimension $r$ in a vector space $V$  (in this paper, we restrict our attention of $V$ to the $d$-dimensional Euclidean space $\mathbb{R}^d$) \cite[Page 21]{lee2009manifolds}. For example, for $r = 1$, the Grassmannian $\text{Gr}(1, d)$ is the space of all lines through the origin in the $d$-dimensional space, and it is the same as the projective space of $d - 1$ dimensions. For the solutions of coupled dynamical system \eqref{eq:coupled_linear_Wij}, they will remain collinear if they start collinearly, but \textit{the line that passes through the solutions of all individual systems may not be identical over time}. In other words, the solutions  will evolve in $\text{Gr}(r, d)$ if they start at a subspace of dimension $r$. {\color{blue} In constant, solutions of the coupled system \eqref{eq:coupled_linear_wij}   with scalar couplings will remain in the same line (or subspace) as spanned by their initial positions.  }

For some typical coupled dynamical systems  with matrix coefficient/couplings reported in the literature that can be described by the general form \eqref{eq:coupled_linear_Wij}, see Table  \ref{table_Wij} in the Appendix.



\section{Applications in convergence analysis for formation  control systems with generalized controllers}  \label{sec:application_formation}
Consider a multi-agent formation control system in the following form
 \begin{equation}\label{eq:position_system}
\dot{x}_i=-\sum_{j\in \mathcal N_i} (\|x_i-x_j\|^2-d_{k_{ij}}^2)(x_i-x_j),
\quad i=1,\ldots,n
\end{equation}
where ${x}_i \in \mathbb{R}^d$ is the position of agent $i$ that lives in $\mathbb{R}^d$, $\mathcal N_i$ denotes agent $i$'s neighboring set, and $d_{k_{ij}}$ is the desired distance that agents $i$ and agent $j$ aim to achieve. In the literature, the above control system \eqref{eq:position_system} is usually called \emph{distance-based formation control system} \cite{oh2015survey}, since the target formation shape is described by a set of interagent distances.

 The collinearity-preserving property for the solutions of the formation control systems \eqref{eq:position_system} was repeatedly observed with different perspectives in several previous papers (e.g., see \cite{krick2009stabilisation,anderson2014counting,oh2014distance}). {\color{blue} In \cite{sun2015rigid}, we have generalized this collinearity-preserving property,  and proved a general dimensional-invariance principle for the formation control system \eqref{eq:position_system}. Inspired by the results in Theorem \ref{thm_rankpreserving_scalar}, }one can also consider the following formation control systems with generalized controllers
 \begin{equation}\label{eq:position_system_modified}
\dot{x}_i=-\sum_{j\in \mathcal N_i} g_{ij}(e_{ij}(x_i, x_j))(x_i-x_j),
\quad i=1,\ldots,n
\end{equation}
where $g_{ij}$ is a continuous function of the distance error $e_{ij}(x_i, x_j)$, which is defined as $e_{ij} = \|x_i-x_j\|^2-d_{k_{ij}}^2$. The local exponential stability of the general formation control system \eqref{eq:position_system_modified} has been discussed in \cite{sun2016exponential}.

The following results are  direct consequences of Theorem~\ref{thm_rankpreserving_scalar} for the general multi-agent formation system \eqref{eq:position_system_modified}.

\begin{corollary}
For 2-D formations, if all the agents start with collinear positions,
they will always be in that collinear subspace spanned by their initial positions under the general  control law described by \eqref{eq:position_system_modified}. Similarly, for 3-D formations, if all the agents
start with coplanar (resp. collinear) positions, then they will
always be in that coplanar (resp. collinear) subspaces spanned by their initial positions under the control law \eqref{eq:position_system_modified}.
\end{corollary}

Conversely, one can also obtain the following dimensional-invariance principle for formation systems \eqref{eq:position_system_modified} with non-collinear/non-coplanar initial positions.
\begin{corollary}
For 2-D/3-D formations, if all the agents start with non-collinear/non-coplanar positions,
they will always be non-collinear/non-coplanar  under the general control law described by \eqref{eq:position_system_modified}.
\end{corollary}

Figures \ref{invariance_collinear_coplanar} and \ref{invariance_noncollinear_noncoplanar} show intuitive explanations of the above two corollaries.
\begin{figure}
  \centering
  \includegraphics[width=80mm]{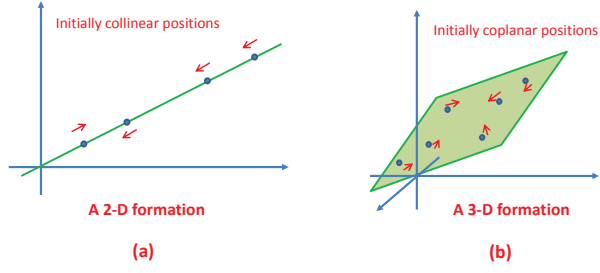}
  \caption{The sets of collinear or coplanar positions are invariant for
2-D/3-D   formation control systems \eqref{eq:position_system_modified}. }
  \label{invariance_collinear_coplanar}
\end{figure}

\begin{figure}
  \centering
  \includegraphics[width=80mm]{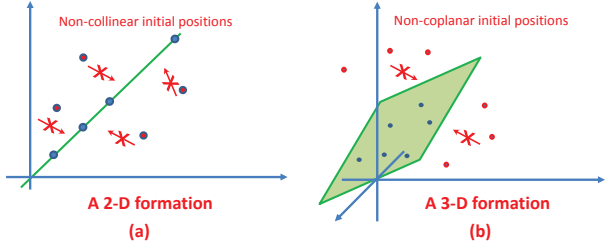}
  \caption{If all agents start at non-collinear (non-coplanar)
positions, then their positions will be non-collinear (non-coplanar) {\color{blue} for any finite time.} }
  \label{invariance_noncollinear_noncoplanar}
\end{figure}
The global analysis of stability and convergence for the formation control system \eqref{eq:position_system} has been discussed in several papers (e.g. \cite{oh2014distance}, \cite{sun2016exponential}, \cite{Xudong2017SIAM}), which turns out to be a very challenging problem. The dimensional-invariance (and subspace-preserving) principles as shown in the above two corollaries will hopefully present additional insights for the convergence and stability analysis of general formation control systems \eqref{eq:position_system_modified}.
In addition,   we can conclude that for any formation control system, if it can be written in the form of \eqref{eq:coupled_linear_wij}, then agents cannot escape collinear/coplanar positions if they start with collinear/coplanar positions. If one needs to design formation controllers to avoid such an invariance property and to enable agents to escape collinear/coplanar positions even if they start collinearly/coplanarly, then one needs to modify the formation  controllers such that they cannot be described by \eqref{eq:coupled_linear_wij}. For typical  examples of formation control systems without the collinear/coplanar invariance property, see \cite{park2012modified,hector2017IFAC}.

\section{Conclusions}  \label{sec:conclusions}
In this paper we {\color{blue} have extended the invariance principles reported in \cite{sun2015rigid} and \cite{collinear2017ACC} to networked coupled systems with general couplings.  A unified analysis via the  rank-preserving matrix flow theory is provided to establish general} invariance principles for coupled dynamical systems (with scalar couplings and with matrix couplings), in relation to the dimensions of the subspaces spanned by their individual solutions.
For coupled dynamical systems with scalar couplings, we prove that their individual solutions satisfy the dimensional-invariance principle (and furthermore, the subspace-preserving principle). For coupled dynamical systems with matrix coefficients/couplings, necessary and sufficient conditions are given to guarantee the dimensional-invariance principle. {\color{blue}The results presented in this paper generalize the findings in \cite{collinear2017ACC} from constant couplings to time-varying and even state-dependent couplings.} The interpretations and implications for the obtained invariance principles are also discussed, with an application to the convergence analysis of formation control systems.
\section*{Appendix I: Background on rank-preserving matrix flow}   \label{sec:appendix}
In this section we will briefly review some background on the rank-preserving flow theory   \cite[Chapter~5]{helmke1994optimization}.

For integers $1 \leq r \leq \text{min}(M,N)$, let
\begin{align}
\mathbb{M}(r, M \times N) = \{X \in \mathbb{R}^{M \times N}  | \text{rank} (X) = r\}
\end{align}
 denote  the set of real $M \times N$ matrices  of fixed rank $r$. The following  results will be useful in later analysis.

 \begin{lemma}
 $\mathbb{M}(r, M \times N)$ is a smooth and connected manifold of dimension $r(M+N-r)$, if $\text{max}(M,N)>1$. The tangent space of  $\mathbb{M}(r, M \times N)$ at an element  $X$ is
  \begin{small}
 \begin{align} \label{eq:tangent_space}
 T_X \mathbb{M} (r, M \times N) = \{\Delta_1 X+ X\Delta_2 | \Delta_1 \in \mathbb{R}^{M \times M}, \Delta_2 \in \mathbb{R}^{N \times N}\}.
 \end{align}
 \end{small}
 \end{lemma}
The proof can be found in \cite[Page~133]{helmke1994optimization}.
A matrix differential equation $\dot X = F(t,X)$ evolving on the matrix space $\mathbb{R}^{M \times N}$ is said to be \emph{rank-preserving} if the rank of every solution  $X(t)$ is constant as a function of $t$, that is, $\text{rank}(X(t)) = \text{rank}(X(0))$ for all $t \geq 0$. The following lemma characterizes such rank-preserving flows (cf. Lemma 1.22 in Chapter 5 of \cite{helmke1994optimization}).

\begin{lemma} \label{lemma_rank}
Let $I  \subset \mathbb{R}$ be an interval and let $A(t) \in \mathbb{R}^{M\times M}$, $B(t) \in \mathbb{R}^{N \times N}$ with $t \in I$   be a continuous time-varying family of matrices. Then
\begin{align} \label{eq:lemma_rank}
\begin{small}
\dot X(t) = A(t)X(t) + X(t)B(t),\,\,\,X(0) \in \mathbb{R}^{M \times N}
\end{small}
\end{align}
is  rank-preserving. Conversely, every rank-preserving differential equation on $\mathbb{R}^{M \times N}$ is of the  form \eqref{eq:lemma_rank} for matrices $A(t)$ and $B(t)$.
\end{lemma}

The proof of Lemma \ref{lemma_rank} is based on the fact that \eqref{eq:lemma_rank} defines a time varying vector field on the subset of the tangent space of $\mathbb{M}(r, M \times N)$ described by  \eqref{eq:tangent_space}.  The full proof can be found in  \cite[Page~139]{helmke1994optimization}. {\color{blue}Note that $I$ can be an open or closed time interval, as long as the solutions of system \eqref{eq:lemma_rank} exist and are well defined over the specified interval. If the existence and uniqueness of solutions to \eqref{eq:lemma_rank} (and the coupled system \eqref{eq:general_system}) are guaranteed for all the time, one can extend the time interval to be any $I \subseteq [0, \infty)$. }
\begin{remark}
The above lemma on rank-preserving flows implies  that the limit value $X(\infty)$ (if it exists) has rank  \emph{less than or equal to}  $\text{rank}(X(0))$. \footnote{ One typical example of $\text{rank}(X(\infty)) < \text{rank}(X(0))$ comes from the formation control problem with unrealizable shapes \cite{sun2014aucc}: If the triangle inequality does not hold for the desired distances in a triangular shape control problem, then all the agents will converge to a \emph{stable collinear equilibrium} for which $\text{rank}(X(\infty)) =1$, even if they start with noncollinear positions with $\text{rank}(X(0)) = 2$. Note that for such flows the rank-preserving property still holds for any finite time but at the limit $t = \infty$ the rank reduces.}  To avoid ambiguity, in this paper we only consider the case that $I$ is a finite time interval. When we say $t \geq 0$, we implicitly exclude the case of $t = \infty$.
\end{remark}

\section*{Appendix II: Extensions on rank-preserving matrix flow}
This section presents some extensions on the rank-preserving flow theory.
The following lemma further characterizes rank-reserving flows on a \textit{symmetric }matrix space. Let $\mathbb{S}(N)$ denote the $N \times N$ real symmetric matrix space. For integers $r \in [1, N]$, let
\begin{align}
\mathbb{S} (r,   N) = \{X \in \mathbb{R}^{N \times N}  | \text{rank} (X) = r\}
\end{align}
 denote  the set of real symmetric $N \times N$  matrices  of fixed rank~$r$.

{\color{blue} For a real symmetric matrix $X \in \mathbb{S} (N)$, its signature, denoted by a pair of three integers $(p, q, s)$, is defined as the numbers of positive, negative and zero eigenvalues (counting multiplicity), respectively.  \footnote{{\color{blue}Note that there is another definition of signature for real symmetric matrices, defined as $\mu =p-q$, i.e., the difference between the number of positive eigenvalues (counting multiplicity) and the number of negative eigenvalues (counting multiplicity). With this definition of matrix signature, the signature-preserving property in Lemma~\ref{lemma_rank_symmetric} still holds. }  } Note that there holds $p+q =r$ for any $X \in \mathbb{S} (r, N)$. A matrix flow $\dot X(t)$ is called \textit{signature-preserving} in $\mathbb{S} (N)$ if the pair $(p, q, s)$ for its solutions $X(t)$ remains constant. }

\begin{lemma} \label{lemma_rank_symmetric}
Let $I  \subset \mathbb{R}$ be an interval and let $A(t) \in \mathbb{R}^{N\times N}$ with $t \in I$   be a continuous time-varying family of matrices. Then
\begin{align} \label{eq:lemma_rank_symmetric}
\dot X(t) = A(t)X(t) + X(t)A^{\top}(t),\,\,\,X(0) \in \mathbb{S}(N)
\end{align}
is a rank-preserving (and hence signature-preserving) flow on $\mathbb{S}(N)$. Conversely, every rank-preserving (and hence signature-preserving) differential equation on $\mathbb{S}(N)$ is of the  form \eqref{eq:lemma_rank_symmetric}.
\end{lemma}

\begin{proof}
The rank-preserving property of $X(t)$ follows from Lemma \ref{lemma_rank} by letting $B(t) = A^{\top}(t)$. The tangent space of  $\mathbb{S} (r,   N)$ at an element  $X$ is
 \begin{align} \label{eq:tangent_space_symmetric}
 T_X \mathbb{S} (r,   N) = \{\Delta   X+ X\Delta^{\top}  | \Delta  \in \mathbb{R}^{N \times N}\}.
 \end{align}
 Therefore \eqref{eq:lemma_rank_symmetric} defines a time varying vector field on each subset of the tangent space of $\mathbb{S} (r,   N)$.  Thus for any initial condition $X(0) \in \mathbb{S}(N)$, the solution $X(t)$ of \eqref{eq:lemma_rank_symmetric} satisfies $X(t) \in \mathbb{S}(N)$, for $t \in I$. {\color{blue}Since the solution $X(t)$ evolves continuously over time, any change of the values $p, q$ will involve a cross-zero scenario or sign change of the corresponding eigenvalues, which will reduce the rank of the symmetric matrix. However, this will violate the rank-preserving property of $X(t)$ and thus it is impossible. Thus,} the signature-preserving property is therefore a direct consequence of the rank-preserving property and the fact that $X(t) \in \mathbb{S}(N)$. Conversely, suppose $X(t)$ is rank-preserving and $X(t) \in \mathbb{S}(N)$ (and therefore is signature-preserving). Then it defines a vector field $F(t, X)$ on $\mathbb{S} (r,   N)$, with $F(t, X) \in T_X \mathbb{S} (r,   N)$ as in \eqref{eq:tangent_space_symmetric}. Letting $\Delta: = A(t)\in \mathbb{R}^{N\times N}$ completes the proof.
\end{proof}

In the following we present a more refined principle, termed subspace-preserving principle, for matrix differential systems.
\begin{lemma} \label{lemma_subspace}
Let $I  \subset \mathbb{R}$ be an interval and let $B(t) \in \mathbb{R}^{N \times N}$ with $t \in I$   be a continuous time-varying family of matrices. Then
\begin{align} \label{eq:lemma_subspace}
\dot X(t) =  X(t)B(t),\,\,\,X(0) \in \mathbb{R}^{M \times N}
\end{align}
is  \emph{subspace-preserving} in the sense that $\text{span}(X(t)) = \text{span}(X(0))$. Conversely, every subspace-preserving differential equation on $\mathbb{R}^{M \times N}$ is of the  form \eqref{eq:lemma_rank} for some matrices   $B(t)$.
\end{lemma}


\begin{proof}
We rewrite   \eqref{eq:lemma_subspace} as $\dot X^{\top}(t) =  B^{\top}(t)X^{\top}(t)$, which has a unique solution given by $X^{\top}(t)  = \Phi_{B^{\top}(t)}(t,0) X^{\top}(0)$, where $\Phi_{B^{\top}(t)}(t,0)$ is the state transition matrix associated with the coefficient matrix $B^{\top}(t)$ (see \cite[Chapter 1.3]{brockett1970finite}).
Therefore, the solution to the system  \eqref{eq:lemma_subspace} can be written as $X(t)  = X(0) \Phi_{B^{\top}(t)}^{\top}(t,0)$. Since the state transition matrix $\Phi_{B^{\top}(t)}^{\top}$ is non-singular \cite[Chapter 1.3]{brockett1970finite}, this implies    that $\text{span}(X(t)) = \text{span}(X(0))$, $\forall t \in I$. For the converse statement, note that $\text{span}(X(t)) = \text{span}(X(0))$ implies that there exists a non-singular  matrix $\Phi$ such that $X(t)  = X(0)\Phi$. In the context of matrix differential equation,   the transpose of the matrix $\Phi$ is the state transition matrix associated with a matrix $B^\top(t)$ in a matrix differential equation in the form of \eqref{eq:lemma_subspace}.
\end{proof}

\begin{remark}
Correspondingly, one can also show that  a matrix differential equation in the form  $\dot X(t) = A(t)  X(t),\,\,\,X(0) \in \mathbb{R}^{M \times N}$, where $A(t) \in \mathbb{R}^{M \times M}$ is a continuous matrix, is \emph{row-subspace-preserving}, in the sense that $\text{span}(X^{\top}(t)) = \text{span}(X^{\top}(0))$. The proof is similar to that of Lemma \ref{lemma_subspace} and is omitted here.
\end{remark}

%
\section*{Appendix III: Proofs of Theorem \ref{thm:rank_preserving_matrix}}
In this section we present the proof for Theorem \ref{thm:rank_preserving_matrix}.

\begin{proof}
Define $X = [x_1, x_2, \ldots, x_n] \in \mathbb{R}^{d \times n}$.
We now determine   conditions for the coefficient/coupling matrix  $W_{ij}$ such that the coupled linear dynamical system \eqref{eq:coupled_linear_Wij} possesses the required dimensional-invariance property.
From Lemma \ref{lemma_rank},
this is equivalent to saying that the matrix differential system $\dot X$ should take the following form
\begin{align} \label{eq:X_equation_compact}
\dot X &= [\dot x_1, \dot x_2, \ldots, \dot x_n] \nonumber \\
&=   A(t) [x_1, x_2, \ldots, x_n] +  [x_1, x_2, \ldots, x_n] B(t)
\end{align}
for some $A(t) \in \mathbb{R}^{d \times d}$ and $B(t) = \{b_{ij(t)}\}  \in \mathbb{R}^{n \times n}$.

Expanding the expression of $\dot X$ in \eqref{eq:X_equation_compact}, one can obtain the equivalent formula in \eqref{eq:X_expanded_1} (in the next page).

Note also that from \eqref{eq:coupled_linear_Wij} the matrix differential  system \eqref{eq:X_equation_compact} can be written as \eqref{eq:X_expanded_2} (in the next page). {\color{blue} Note that to keep a short display of the equations we have suppressed the expression of time $t$ in  Eqs.~\eqref{eq:X_expanded_1} and \eqref{eq:X_expanded_2}, but the matrices $A$ and $B$ can be time-varying. }  In order to guarantee the dimensional-invariance principle, each coefficient term in the system  \eqref{eq:X_expanded_2} should take the identical form as in \eqref{eq:X_expanded_1},
which implies
\begin{align}
W_{11}(t) &= A(t) + b_{11}(t) I_d, \nonumber \\
W_{12}(t) &= b_{21}(t) I_d, \nonumber \\
W_{13}(t) &= b_{31}(t) I_d, \nonumber \\
&\vdots \nonumber \\
W_{ii}(t) &= A(t) + b_{ii}(t) I_d, \nonumber \\
W_{ij}(t) &= b_{ji}(t) I_d, \,\,\, i, j =  1, 2, \ldots, n, i \neq j, \nonumber \\
&\vdots  \nonumber \\
W_{nn}(t) &= A(t) + b_{nn}(t) I_d,
\end{align}
which is the necessary and sufficient condition to guarantee the dimensional-invariance property for the coupled dynamical system \eqref{eq:coupled_linear_Wij}.
\end{proof}

\section*{Appendix IV: A brief review of coupled systems that can be described by  \eqref{eq:coupled_linear_wij} and \eqref{eq:coupled_linear_Wij}}

We review and summarize  in   Table \ref{table_wij} and Table \ref{table_Wij}  several popular coupled dynamical systems reported in the vast literature, which  can be described by \eqref{eq:coupled_linear_wij} and \eqref{eq:coupled_linear_Wij}, respectively. As a consequence of Theorem \ref{thm_rankpreserving_scalar} and Corollary \ref{corollary_subspace_preserving}, for all the coupled or networked control systems with scalar couplings reviewed in Table \ref{table_wij}, dimensional-invariance (and furthermore, subspace-preserving) principles  are guaranteed. For coupled/networked control systems reviewed in Table \ref{table_Wij}, if the matrix condition in Theorem \ref{thm:rank_preserving_matrix} is satisfied, then they also possess the dimensional-invariance property. For example, for the synchronization control of \emph{identical} networked linear systems with matrix coefficients/couplings \cite{scardovi2009synchronization} (i.e., for the Type I system, with $A_i: = A, \forall i$ and $W_{ij} = b_{ij} I_d$), the matrix condition of Theorem \ref{thm:rank_preserving_matrix} is satisfied and the solutions of such networked control systems possess the dimensional-invariance principle. In contrast, for the Type II coupled systems for linear system synchronization, the condition \eqref{eq:theorem2_condition2} in Theorem \ref{thm:rank_preserving_matrix} would be violated and thus the dimensional-invariance property is not guaranteed.

\begin{figure*}[ht]
\begin{align}  \label{eq:X_expanded_1}
\dot X
 & =  A [x_1, x_2, \ldots, x_n] +  [x_1, x_2, \ldots, x_n] \left[
\begin{array}{cccc}
b_{11} & b_{12} & \dots & b_{1n} \\
b_{21} & b_{22} & \dots & b_{2n} \\
\vdots & \vdots & \ddots & \vdots \\
b_{n1} & b_{n2} & \dots & b_{nn} \\
\end{array}
\right] \nonumber \\
&= [(A + b_{11} I_d)x_1 + b_{21} x_2 + \cdots + b_{n1} x_n, b_{12}x_1 + (A + b_{22}I_d) x_2 + \cdots + b_{n2} x_n, \ldots,   b_{1n}x_1 + b_{2n} x_2 + \cdots + (A+b_{nn}I_d) x_n]
\end{align}
\begin{align}  \label{eq:X_expanded_2}
\dot X &= [\dot x_1, \dot x_2, \ldots, \dot x_n]  = [W_{11} x_1 + W_{12} x_2 + \cdots + W_{1n} x_n,  W_{22} x_2 + W_{21} x_1 + \cdots + W_{2n} x_n, \ldots, W_{n1} x_1 + W_{n2} x_2 + \cdots +  W_{nn} x_n]
\end{align}
\end{figure*}

\begin{table*}[t]
\caption{Coupled dynamical systems that can be described by  \eqref{eq:coupled_linear_wij}}
\begin{small}
\begin{tabular}{l*{4}{c}r}
\hline \hline
  Ref.          & Coupled/networked control systems    & System dynamics equation & Coefficient/coupling term     \\
\hline
\cite{costello2014degree}, \cite{costello2015global}, etc. & Network distributed computation
& $\begin{aligned}[t]
\dot{x}_i  &=\sum_{j\in \mathcal N_i} w_{ij}x_j, \quad i=1,\ldots,n \nonumber \\
\end{aligned} $
 & $w_{ij}$         \\

\hline
\cite{olfati2004consensus}, \cite{ren2007information}, \cite{martin2013continuous}, etc. & Multi-agent consensus
& $\begin{aligned}[t]
\dot{x}_i  &=\sum_{j\in \mathcal N_i} a_{ij}(x_j-x_i), \quad i=1,\ldots,n \nonumber \\
& a_{ij}: \text{weighted adjacency matrix} \nonumber \\ &\text{including undirected/directed graphs,} \nonumber \\ &
\text{static/time-varying/switching topologies} \nonumber \\
\end{aligned}$
& $\begin{aligned}[t]  w_{ij} &= a_{ij}, i \neq j, \nonumber \\  w_{ii} &= -\sum_{j\in \mathcal N_i} a_{ij} \end{aligned} $        \\

\hline
\cite{anderson2014counting}, \cite{krick2009stabilisation},  \cite{sun2016exponential}, etc.& \begin{tabular}{@{}c@{}} Distance-based formation \\ shape control\end{tabular}
& $\begin{aligned}[t]
\dot{x}_i  &=-\sum_{j\in \mathcal N_i} g_{ij}(x_i-x_j), \quad i=1,\ldots,n \nonumber \\
e_{ij}  &= \|x_i-x_j\|^2 - d_{ij}^2,  \nonumber \\
g_{ij}   &= e_{ij}, \text{or}  \,\,\,g_{ij} = g_{ij}(e_{ij}) \nonumber \\
\end{aligned} $
& $\begin{aligned}[t]  w_{ij} &= g_{ij}, i \neq j, \nonumber \\  w_{ii} &= -\sum_{j\in \mathcal N_i} g_{ij} \end{aligned} $  \\

\hline
\cite{chen2016swarm}, \cite{chen2015controllability}, \cite{lin2016necessary}, etc. &   \begin{tabular}{@{}c@{}} Multi-agent coordination/ \\ affine formation/swarming\end{tabular}
& $\begin{aligned}[t]
\dot{x}_i  =-\sum_{j\in \mathcal N_i} u_{ij}(t,x)(x_i-x_j), \quad i=1,\ldots,n \nonumber \\
\end{aligned} $
& $\begin{aligned}[t]  w_{ij} &= u_{ij}, i \neq j, \nonumber \\  w_{ii} &= -\sum_{j\in \mathcal N_i} u_{ij} \end{aligned} $  \\

\hline

\cite{setter2016trust}, \cite{setter2017trust}, etc.& \begin{tabular}{@{}c@{}} {\color{blue}  Multi-agent networks with} \\ {\color{blue}trust-based interactions }\end{tabular}  
& $\begin{aligned}[t]   
\dot{x}_i  &=\sum_{j\in \mathcal N_i} \tau_{ij}(x_j-x_i), \quad i=1,\ldots,n \nonumber \\
\dot \tau_{ij}  &= -\frac{\partial F_{ij}(\|x_i - x_j\|)}{\partial x_j} \dot x_j,  \nonumber \\
\text{or} \,\, \,\, \dot{x}_i  &= \tau_{i}\sum_{j\in \mathcal N_i} (x_j-x_i), \quad i=1,\ldots,n \nonumber \\
\dot \tau_{i}  &= -\sum_{j \in N_i}\frac{\partial F_{ij}(\|x_i - x_j\|)}{\partial x_j} \dot x_j \nonumber \\\nonumber \\  
\end{aligned} $
& $\begin{aligned}[t]  w_{ij} &= \tau_{ij}, i \neq j, \nonumber \\  w_{ii} &= -\sum_{j\in \mathcal N_i} \tau_{ij}     \\ 
\text{or} \,\, \,\, w_{ij} &= \tau_{i}, (i, j) \in \mathcal E, \nonumber \\  
w_{ij} &= 0, (i, j) \notin \mathcal E, \nonumber \\
w_{ii} &= -\sum_{j\in \mathcal N_i} \tau_{i} \end{aligned} $  \\ 
 
\hline
\end{tabular}
\end{small}
  \label{table_wij}
\end{table*}

\begin{table*}[t]
\caption{Coupled dynamical systems that can be described by  \eqref{eq:coupled_linear_Wij}}
\begin{small}
\begin{tabular}{l*{4}{c}r}
\hline \hline
  Ref.          & Coupled/networked control systems    & System dynamics equation & Coefficient/coupling term     \\
\hline
\cite{collinear2017ACC}  & Collinear dynamical systems
& $\begin{aligned}[t]
\dot{x}_i  &=\sum_{j=1}^n A_{ij}x_j, \quad i=1,\ldots,n \nonumber \\
\end{aligned} $
 & $W_{ij} = A_{ij}$         \\

\hline
\cite{tuna2016synchronization}, \cite{trinh2017theory}, etc.   & Matrix-weighted consensus
& $\begin{aligned}[t]
\dot{x}_i  &=\sum_{j\in \mathcal N_i} Q_{ij}(x_j-x_i), \quad i=1,\ldots,n \nonumber \\
& Q_{ij} \in \mathbb{R}^{d \times d}   \nonumber \\
\end{aligned}$
& $\begin{aligned}[t]  W_{ij} &= Q_{ij}, i \neq j, \nonumber \\  W_{ii} &= -\sum_{j\in \mathcal N_i} Q_{ij} \end{aligned} $        \\

\hline
\cite{scardovi2009synchronization}, \cite{li2010consensus}, \cite{wieland2011internal}, etc.  &   \begin{tabular}{@{}c@{}} Synchronization control  \\ for networked linear systems \end{tabular}
& $\begin{aligned}[t]
&  \text{Type I:} \,\,\, \dot{x}_i   =A_i x_i -\sum_{j\in \mathcal N_i} b_{ij}(t)(x_i-x_j)  \nonumber \\
 & \quad i=1,\ldots,n,  A_i \in \mathbb{R}^{d \times d}, b_{ij}(t) \in \mathbb{R}  \nonumber \\
 \hline
&  \text{Type II:} \,\,\, \dot{x}_i   =A_i x_i -\sum_{j\in \mathcal N_i} B_{ij}(t)(x_i-x_j)  \nonumber \\
&  \quad i=1,\ldots,n,  \,\,\,\,\, A_i, B_{ij}(t) \in \mathbb{R}^{d \times d} \nonumber \\
\end{aligned} $
& $\begin{aligned}[t]  W_{ij} &= b_{ij}I_d, i \neq j, \nonumber \\  W_{ii} &= A_i -\sum_{j\in \mathcal N_i} b_{ij}I_d \nonumber \\
\hline
 W_{ij} &= B_{ij}, i \neq j, \nonumber \\  W_{ii} &= A_i-\sum_{j\in \mathcal N_i} B_{ij} \end{aligned} $  \\

\hline
\end{tabular}
\end{small}
  \label{table_Wij}
\end{table*}

\section*{ACKLONOGEMENT}
We thank the Associate Editor and the  anonymous reviewers for their constructive suggestions which help improve the presentation and clarity of this paper. The first author would like to thank the late Professor Uwe Helmke for his valuable and insightful discussions   on this topic.  This paper is dedicated to the memory of him. 

\bibliography{2017_Dimension_invariance}
\bibliographystyle{ieeetr}

\end{document}